\documentclass[aps,prl,superscriptaddress,showpacs]{revtex4}

\usepackage{amsfonts}
\usepackage{graphicx}
\usepackage{amsmath}
\usepackage{amsfonts,amssymb}
\usepackage{color}
\usepackage{slashed}
\usepackage{bm}

\def\endprf{\hfill  {\vrule height6pt width6pt depth0pt}\medskip}
\newenvironment{proof}{\noindent {\bf Proof} }{\endprf\par}

\begin{document}


\title{Chiral Wess-Zumino model and breaking of supersymmetry}


\author{Marco Frasca}
\email[]{marcofrasca@mclink.it}
\affiliation{Via Erasmo Gattamelata, 3 \\ 00176 Roma (Italy)}


\date{\today}

\begin{abstract}
We prove that the Wess-Zumino model breaks supersymmetry in the low-energy limit. The mechanism is identical to the breaking of chiral symmetry in quantum chromodynamics as the low-energy limit is a Nambu-Jona-Lasinio model and the gap equation has a solution. This happens at a critical coupling at lower momenta. Instanton liquid plays an essential role in the behavior of the theory in this limit.
\end{abstract}

\pacs{11.30.Pb,11.30.Rd, 12.60.Jv,12.38.Aw}

\maketitle

\section{Introduction}

Supersymmetry (SUSY) is a hypothetical symmetry relating bosonic and fermionic matter. It has been successful for some theoretical aspects as the solution to the hierarchy problem of the Standard Model and the unification of the running coupling for strong and electroweak interactions. This symmetry must be broken as, otherwise, we would have observed partners for each known particle. This means that, if this symmetry is realized in nature, some mechanism exists that breaks it making such superpartners very massive having not been seen yet at accelerator facilities.

There has been some proposals for a breaking mechanism. Spontaneous breaking has been proposed by Fayet and Iliopolous and by O'Raifeartaigh \cite{Fayet:1974jb,O'Raifeartaigh:1975pr}. In this case, auxiliary fields, $F$ or $D$, have a non-null vacuum expectation value. But these approaches do not apply to the minimal supersymmetric standard model (MSSM) and possible extensions. MSSM was firstly proposed by Georgi and Dimopolous \cite{Dimopoulos:1981zb}. In order to overcome the difficulties with breaking of supersymmetry and obtain low-energy phenomenology, they adopted ``softly'' broken SUSY. The idea is to use explicitly breaking terms in the Lagrangian of the model but granting renormalizability and invariance under electroweak symmetry group. Here the situation is more favorable with respect the Standard Model where terms like these will spoil renormalizability. Of course, this just moves the problem elsewhere and one has to devise a supersymmetry breaking mechanism that would give rise to a standard Higgs mechanisms at lower energies.

Supersymmetry cannot accomodate the Higgs mechanism. This is known as the $\mu$ problem \cite{Aitchison:2007fn} and means that a mass term with a negative sign is not possible. Moreover, if this is not possible at tree level, non-renormalization theorems \cite{Ferrara:1974fv,Wess:1973kz,Grisaru:1979wc} grant that this never happens. So, Witten, in a pioneering work, proposed to consider dynamical symmetry breaking \cite{Witten:1981nf} (a review is \cite{Shadmi:1999jy}). In this case, supersymmetry is kept but dynamical effects give different masses to the fields.

The paradigm of the dynamical symmetry breaking is quantum chromodynamics (QCD). Recent works \cite{Frasca:2011bd,Kondo:2010ts,Frasca:2008zp} have shown that the low-energy limit of QCD is a non-local Nambu-Jona-Lasinio model \cite{Nambu:1961tp,Nambu:1961fr,Klevansky:1992qe}. This model shows a dynamical breaking of chiral symmetry and an effective mass is obtained with a gap equation at a critical coupling. Besides, this theory has a rich set of bound states that describes very well the observed phenomenology of hadronic physics at lower energies \cite{Ebert:1994mf}. Quite recently, we showed how massive solutions can emerge for a massless self-interacting scalar field both at classical and quantum level \cite{Frasca:2009bc}. Wess-Zumino model has all the characteristics of these models and so it can be the starting point for a more general proof of low-energy breaking of supersymmetry through field dynamics. This aspect has never been discussed in literature so far.

In this letter we will present a proof of dynamical breaking of supersymmetry arising from the low-energy dynamics of the Wess-Zumino model. We will derive the critical value of the coupling for which symmetry breaking appears showing how this mechanism is hidden in the model.

\section{A theorem on the Yukawa model\label{sec2}}

In order to give a simple example of the way massless interacting fields can give rise to mass, we consider a very simple model: A scalar field just interacting with a fermion well-known in literature as Yukawa model. So, we have
\begin{equation}
\label{eq:yuk}
    {\cal L} = \bar q\left[{\slashed\partial}-g\sigma\right]q
    +\frac{1}{2}(\partial\sigma)^2-\frac{1}{2}\mu^2\sigma^2.
\end{equation}
For this model we can prove the following:

\newtheorem{yuk}{Theorem}
\begin{yuk}
A Yukawa model without a self-interaction term and interacting with a fermion field displays breaking of chiral symmetry.
\end{yuk}

\begin{proof} 
From the Lagrangian (\ref{eq:yuk}) one has the following generating functional
\begin{equation}
    Z[j,\eta,\bar\eta]=N\int[d\sigma][dq][d\bar q]
    \exp\left\{i\int d^4x\left[{\cal L}+j\sigma+\bar\eta q-\bar q\eta\right]\right\}.
\end{equation}
We can rewrite this like
\begin{eqnarray}
    Z[j,\eta,\bar\eta]&=&N\int[dq][d\bar q]
    \exp\left\{i\int d^4x\left[\bar q{\slashed\partial}q+\bar\eta q
    -\bar q\eta\right]\right\}\times \\ \nonumber
    &&\int[d\sigma]
    \exp\left\{i\int d^4x\left[\frac{1}{2}\sigma(-\partial^2-\mu^2)\sigma
    -g\bar qq\sigma+j\sigma\right]\right\}.
\end{eqnarray}
Now, we change integration variable in the last integral writing down
\begin{equation}
  \sigma=\sigma'-\int d^4y\Delta(x-y)j(y)+g\int d^4y\Delta(x-y)\bar q(y)q(y)
\end{equation}
being $\Delta(x-y)$ the Green function for the scalar field. One gets immediately
\begin{eqnarray}
    Z[j,\eta,\bar\eta]&=&
    N\exp\left\{\frac{i}{2}\int d^4xd^4yj(x)\Delta(x-y)j(y)\right\} \\ \nonumber    
    &&\int[dq][d\bar q]
    \exp\left\{i\int d^4x
    \left[\bar q{\slashed\partial}q+\bar\eta q-\bar q\eta\right]\right\} \\ \nonumber
    &&\exp\left\{\frac{i}{2}g^2
    \int d^4xd^4y\bar q(x)q(x)\Delta(x-y)\bar q(y)q(y)\right\} \\ \nonumber
    &&\exp\left\{\frac{i}{2}g\int d^4xd^4yj(x)\Delta(x-y)\bar q(y)q(y)\right\} \\ \nonumber
    &&\exp\left\{\frac{i}{2}g\int d^4xd^4y\bar q(x)q(x)\Delta(x-y)j(y)\right\}.
\end{eqnarray}
For $g$ large enough and smaller momenta this is just a Nambu-Jona-Lasinio model that breaks chiral symmetry.
\end{proof}

One can also prove the reverse by bosonization of the Nambu-Jona-Lasinio model. In this case, at the leading order one gets scalar fields with massive (quadratic) terms. Next order corrections shift the vacuum value of the scalar field from zero again breaking chiral symmetry. But these results are well known \cite{Ebert:1994mf}.

\section{Infrared limit of a self-interacting scalar field\label{sec3}}

Infrared limit of a massless scalar field with a quartic interaction has an infrared trivial fixed point \cite{Frasca:2011bd,Frasca:2011pz,Frasca:2010ce}. This can be seen in the following way. Consider the generating functional
\begin{equation}
   Z[j]=\int [d\phi]e^{i\int d^4x\left(\frac{1}{2}(\partial\phi)^2-\frac{\lambda}{4}\phi^4+j\phi\right)}.
\end{equation}
We can rescale space-time variables as $x\rightarrow\sqrt{\lambda}x$ and will get
\begin{equation}
    Z[j]=\int[d\phi]\exp\left[i\frac{1}{\lambda}\int d^4x\left(\frac{1}{2}(\partial\phi)^2-\frac{1}{4}\phi^4+\frac{1}{\lambda}j\phi\right)\right].
\end{equation}
Being $j$ arbitrary, we can rescale it as $j\rightarrow j/\lambda$ and then, we take $\phi=\sum_{n=0}^\infty\lambda^{-n}\phi_n$ with $\phi_0$ such that 
\begin{equation}
\label{eq:phi4}
\partial^2\phi_0+\phi_0^3=j.
\end{equation} 
So,
\begin{equation}
    Z[j]=\exp\left[i\int d^4x\left(\frac{1}{2}(\partial\phi_0)^2-\frac{\lambda}{4}\phi_0^4+j\phi_0\right)\right]
    \int[d\phi_1]\exp\left[i\frac{1}{\lambda}\int d^4x\left(\frac{1}{2}(\partial\phi_1)^2-\frac{3}{2}\lambda\phi_0^2\phi_1^2\right)+O\left(\frac{1}{\lambda^2}\right)\right]
\end{equation}
after undoing the rescaling. This expansion holds in the formal limit $\lambda\rightarrow\infty$. Now, the equation (\ref{eq:phi4}) admits the solution \cite{Frasca:2011pz,Frasca:2010ce}
\begin{equation}
   \phi_0=\mu\int d^4x'\Delta(x-x')j(x')+\mu^2\lambda\int d^4x'd^4x''\Delta(x-x')[\Delta(x'-x'')]^3j(x')+O(j^3).
\end{equation}
$\mu$ is an integration constant having the dimension of a mass. This is a current expansion provided the propagator solves $\partial^2\Delta(x-x')+\lambda\Delta^3(x-x')=\delta^4(x-x')$. The propagator has a closed analytical form given by in momenta space
\begin{equation}
\label{eq:green}
   \Delta(p)=\sum_{n=0}^\infty(2n+1)\frac{\pi^2}{K^2(i)}\frac{(-1)^{n}e^{-(n+\frac{1}{2})\pi}}{1+e^{-(2n+1)\pi}}
   \frac{1}{p^2-m_n^2+i\epsilon}
\end{equation}
being $m_n=(2n+1)(\pi/2K(i))\left(\lambda/2\right)^{\frac{1}{4}}\mu$ and $K(i)\approx 1.3111028777$ an elliptic integral. So, the leading contribution to the generating functional is just a Gaussian one and the theory displays a trivial behavior at the infrared limit of momenta going to zero. We just notice that the spectrum of the theory is that of free massive particles with a superimposed spectrum of a harmonic oscillator. For the sake of completeness we rewrite here the generating functional as
\begin{equation}
   Z[j]\approx \exp\left[\int d^4xd^4yj(x)\Delta(x-y)j(y)\right]
   \int[d\phi_1]\exp\left[i\frac{1}{\lambda}\int d^4x\left(\frac{1}{2}(\partial\phi_1)^2-\frac{3}{2}\lambda\mu^2\left[\int d^4y\Delta(x-y)j(y)\right]^2\phi_1^2\right)\right].
\end{equation}
For our aims, in the following we will neglect the next-to-leading order correction and working just at the fixed point.

\section{Infrared limit of the Wess-Zumino model\label{sec4}}

Massless Wess-Zumino model has the Lagrangian \cite{Weinberg:2000cr}
\begin{eqnarray}
     L &=& \frac{1}{2}(\partial A)^2+\frac{1}{2}(\partial B)^2+\frac{1}{2}\bar\psi i\slashed\partial\psi \nonumber \\
       &&-\frac{1}{2}g^2(A^2+B^2)^2-g(\bar\psi\psi A+i\bar\psi\gamma^5\psi B)
\end{eqnarray}
being $\psi$ a Majorana field, $A=A^\dagger$ a scalar field and $B=B^\dagger$ a pseudo-scalar field. This Lagrangian is invariant under supersymmetric variations that can be stated in the form
\begin{eqnarray}
    \delta A(x)&=&\bar\epsilon\psi(x) \nonumber \\
    \delta B(x)&=&-i\bar\epsilon\gamma^5\psi(x) \nonumber \\
    \delta\psi(x) &=& \partial_\mu(A+i\gamma^5B)\gamma^\mu\epsilon.
\end{eqnarray}
These hold on-shell because we have already removed the auxiliary fields using the equations of motion. We note that this model, at small coupling $g\rightarrow 0$, has no chiral symmetry breaking and supersymmetry is preserved. Now, we assume that the coupling is in the opposite limit $g\rightarrow\infty$ in order to apply the analysis given in sec.\ref{sec4}. We will see that thee exists a critical coupling granting the breaking of supersymmetry. In order to reach the trivial infrared fixed point, we need to solve the classical equations of motion for the scalar fields
\begin{eqnarray}
\label{eq:susy}
    \partial^2A+2g^2A(A^2+B^2)&=&-g\bar\psi\psi \nonumber \\
    \partial^2B+2g^2B(A^2+B^2)&=&-ig\bar\psi\gamma^5\psi.
\end{eqnarray}
So, to stay at the fixed point we have to evaluate the Green functions
\begin{eqnarray}
    \partial^2\Delta_A+2g^2\Delta_A(\Delta_A^2+\Delta_B^2)&=&\delta^4(x) \nonumber \\
    \partial^2\Delta_B+2g^2\Delta_B(\Delta_A^2+\Delta_B^2)&=&\delta^4(x).
\end{eqnarray}
It is not difficult to recognize that a solution is obtained through the condition $\Delta_A=\Delta_B=\Delta$ and we recover in full the analysis of sec.\ref{sec4} provided we assume $\lambda=4g^2$ and the propagator is the one given in eq.(\ref{eq:green}). Further, we need to show that this solves eqs.(\ref{eq:susy}). So, we write
\begin{eqnarray}
    A&=&-g\mu\int d^4y\Delta(x-y)\bar\psi(y)\psi(y)+\delta A \nonumber \\
    B&=&-ig\mu\int d^4y\Delta(x-y)\bar\psi(y)\gamma^5\psi(y)+\delta B.
\end{eqnarray}
By direct substitution and using the equation for $\Delta(x-y)$ we get the exact equations
\begin{eqnarray}
\label{eq:susy1}
    \partial^2\delta A+2g^2\delta A(\delta A^2+\delta B^2)&=&-2g^3\mu\int d^4y\Delta^3(x-y)\bar\psi(y)\psi(y) \nonumber \\
    && -2g^2[A_0(A_0^2+B_0^2)+(3A_0^2+B_0^2)\delta A+A_0(3\delta A^2+\delta B^2)+2(A_0B_0+B_0\delta A)\delta B] \nonumber \\
    \partial^2\delta B+2g^2\delta B(\delta A^2+\delta B^2)&=&-2ig^3\mu\int d^4y\Delta^3(x-y)\bar\psi(y)\gamma^5\psi(y) \nonumber \\
    && -2g^2[B_0(A_0^2+B_0^2)+(3A_0^2+B_0^2)\delta B+B_0(3\delta B^2+\delta A^2) \nonumber \\
    &&+2(A_0B_0+A_0\delta B)\delta A]
\end{eqnarray}
where we have set
\begin{eqnarray}
    A_0&=&-g\mu\int d^4y\Delta(x-y)\bar\psi(y)\psi(y) \nonumber \\
    B_0&=&-ig\mu\int d^4y\Delta(x-y)\bar\psi(y)\gamma^5\psi(y).
\end{eqnarray}
These equations can be solved iteratively and the next-to-leading order correction is obtained neglecting higher powers of the currents. This makes the argument consistent. In this way, we can immediately recover a non-local Nambu-Jona-Lasinio model \cite{Frasca:2011bd}
\begin{eqnarray}
   S_{WZ}&=&\frac{1}{2}\int d^4xd^4yj(x)\Delta(x-y)j(y)
   +\frac{1}{2}\int d^4x\bar \psi(x)i\slashed{\partial}\psi(x). \nonumber \\
   &&+g^2\mu\int d^4xd^4y'\Delta(x-y')\left[\bar\psi(x)\psi(x)\bar\psi(y')\psi(y')+ \bar\psi(x)i\gamma^5\psi(x)\bar\psi(y')i\gamma^5\psi(y')\right]\nonumber \\
   &&+O(1/2g).
\end{eqnarray}
Now, we consider the form factor for this model. Using eq.(\ref{eq:green}) one has in momenta space
\begin{equation}
   {\cal G}(p)=-g^2\sum_{n=0}^\infty\frac{B_n}{p^2-(2n+1)^2(\pi/2K(i))^2\sigma+i\epsilon}=\frac{G}{2}{\cal C}(p)
\end{equation}
being $G$ the coupling of the local Nambu-Jona-Lasinio model and $G=2{\cal G}(0)=(2g^2/\sigma)\sum_{n=0}^\infty\frac{B_n}{(2n+1)^2(\pi/2K(i))^2}\approx 1.57(g^2/\sigma)$, having set $\sqrt{\sigma}=\left(2g^2\right)^{\frac{1}{4}}\mu$. We have chosen a normalization such that ${\cal C}(0)=1$. This is the form factor proper to an instanton liquid \cite{Schafer:1996wv} and is the reason for the Majorana field to get a different mass from the scalar field breaking chiral symmetry and so supersymmetry.

We can write down the gap equation in this case giving \cite{Frasca:2011bd}
\begin{equation}
    M(p)={\cal C}(p)v
\end{equation}
and
\begin{equation}
    v=4G\int\frac{d^4p}{(2\pi)^4}{\cal C}(p)\frac{M(p)}
    {p^2+M^2(p)}
\end{equation}
that provides the effective mass for the Majorana field. For very small momenta ${\cal C}(p)\approx 1$ and this is just the gap equation already seen in a local Nambu-Jona-Lasinio model. This admits a simple solution given in \cite{Klevansky:1992qe}: There is a critical coupling $G_c\mu^2=\pi^2$ granting the breaking of chiral symmetry and so, the breaking of supersymmetry at lower energies as the fermion acquires a mass different from the one of the two scalar fields. This complete the proof.

\section{Conclusions\label{sec5}}

In this letter we have shown that the Wess-Zumino model displays dynamical symmetry breaking at lower energies and supersymmetry is broken. This is so because, in this limit, we recover a Nambu-Jona-Lasinio model. This means that exists a critical coupling beyond which particles get different masses. This should open the perspective to improve the MSSM imrpoving its possibility to agree with observed data. This also means that supersymmetry has inside itself the proper mechanism for breaking that becomes effective as energy lowers and this appears manifestly in the non-linear terms that are needed to preserve such an invariance.


\end{document}